\documentclass[12pt]{article}

\usepackage{amsfonts} 
\usepackage{amssymb} 
\usepackage{amsmath}
\usepackage{mathrsfs}
\usepackage{eucal}
\usepackage[dvips]{graphicx} 
\usepackage{makeidx}
\usepackage{color}
\usepackage{url}
\usepackage[T1]{fontenc}
\usepackage[utf8]{inputenc}
\usepackage[italian]{babel}

\begin{document}

\newenvironment {proof}{{\noindent\bf Proof.}}{\hfill $\Box$ \medskip}

\newtheorem{defi}{Definition}[section]
\newtheorem{pro}{Proposition}[section]
\newtheorem{teo}{Theorem}[section]
\newtheorem{lemma}{Lemma}[section]
\newtheorem{coro}{Corollary}[section]
\newtheorem{rem}{Remark}[section]
\newtheorem{cond}{Condition}[section]
\newtheorem{ass}{Assumption}[section]

\renewcommand {\theequation}{\arabic{section}.\arabic{equation}}

\newcommand{\red}[1]{\textcolor{red}{#1}}

\def \non{{\nonumber}}
\def \hat{\widehat}
\def \tilde{\widetilde}
\def \bar{\overline}
\def \P{\mathbb{P}}
\def \E{\mathbb{E}}
\def \R{\mathbb{R}}
\def \Z{\mathbb{Z}}
\def \1{\mathbf 1}

\title{\large {\bf Risk Neutral Valuation of Inflation-Linked \\Interest Rate Derivatives}}
                                                       
\author {F. Antonacci$^{1}$, C. Costantini$^{2}$, F. D'Ippoliti$^{3}$, M. Papi$^{4}$}


\date{}

\maketitle

\vspace*{-1cm}

\begin{center}
$^{1}$ Università di Chieti-Pescara, Pescara (Italy)\\
Dipartimento Economics\\
e-mail: flavia.antonacci@unich.it\\[1.5mm]

$^{2}$ Università di Chieti-Pescara, Pescara (Italy)\\
Dipartimento Economics\\
e-mail: c.costantini@unich.it\\[1.5mm]

$^{3}$  RBC Capital Markets, London (UK)\\
Department of Electrical and Information Engineering\\[1.5mm]

$^{4}$ UCBM, Rome (Italy)\\
School of Engineering\\
e-mail: m.papi@unicampus.it\\[1.5mm]
\end{center}

\begin{abstract}
We propose a model for the joint evolution 
of European inflation,
the European Central Bank official interest rate and the short-term interest
rate, in a stochastic, continuous time setting. 

We derive the valuation equation for a contingent claim 
depending potentially on all three factors. This 
valuation equation reduces to a 
finite number of Cauchy problems for a degenerate parabolic 
PDE with non-local terms. We show that the price of 
the contingent claim is the 
only viscosity solution of the valuation equation. 

We also provide an efficient numerical scheme to compute the 
price and implement it in an example. 
\end{abstract}







\begin{table}[b]
{\small
{\bf ------------------------------------------------------}\\
Corresponding author: Flavia Antonacci\\
Department of Economics,
Università di Chieti-Pescara,\\\ Viale Pindaro, 42 - 65127 Pescara (Italy)\\
              Ph.: +39-08-0854537938\\
              e-mail: flavia.antonacci@unich.it
}

\end{table}


\setcounter{equation}{0}
\section{Introduction}

The issuance of sovereign bonds linked to euro area
inflation began with the introduction of bonds indexed to the French
consumer price index (CPI) excluding tobacco (Obligations
Assimilables du Tr\'{e}sor index\'{e}es or OATis) in 1998. In 2003,
Greece, Italy and Germany decided to issue inflation linked 
bonds too. Despite that, the inflation linked derivatives market is only on
its infancy. Some typical examples are \textit{inflation 
caps}, which
pay out if the inflation exceeds a certain threshold over a given
period, or \textit{inflation protected annuities} that guarantee a
real rate of return 
at or above inflation (for a list of inflation
derivatives, see e.g. Hughston~\cite{Hughston1998}).

The pricing of inflation linked derivatives is related to both interest
rate and foreign exchange theory. In their seminal work 
of 2003, Jarrow and
Yildirim~\cite{JarrowYildirim2003} proposed an approach based
on foreign currency and interest rate derivatives valuation. On the
other hand, there is some empirical and theoretical evidence
that bond prices, inflation, interest rates, monetary policy and
output growth are related. In particular, both inflation 
and interest rates 
are clearly related to the activity of central banks. 

In the present work we propose a model for the joint 
evolution of European inflation, the European Central Bank
(henceforth ECB) official interest rate and the short-term
interest rate, and use it to price European type 
derivatives whose payoff depends potentially on all three factors. 
To the best of our knowledge, ours is the 
first model that takes into account the interaction 
among all these three factors. 
With the 2007-2008 financial crisis it has become clear 
that there is another risk factor underlying bond prices, 
namely credit risk, but we leave the construction 
of a model that incorporates this factor for 
future work. 

Our model is a stochastic, continuous time one. 
More precisely, the ECB interest rate evolves as a pure jump process 
with intensity that depends both on its current value 
and on the current value of inflation. See Section 
\ref{subsecECB} for more detail.
Inflation is modeled as a piecewise constant process that jumps 
at fixed times $t_i$: This reflects the fact that inflation is 
measured at regular times. 
The new value at $t_i$ is given 
by a Gaussian random variable with expectation depending 
on the previous value of inflation and on the current 
values of both the ECB and short-term interest rates. 
Gaussian distributions for real and nominal interest rates 
are employed, for instance, in Mercurio \cite{Mercurio2005}. 
Finally, the short-term interest rate follows a CIR type model 
with reversion towards an affine function of the ECB 
interest rate and diffusion coefficient depending on the spread 
between itself and the ECB interest rate. 

Many models proposed to price inflation-indexed
derivatives fall in the class of affine models (see e.g. 
Ho, Huang and Yildirim \cite{HoHuangYildirim} and 
Waldenberger \cite{Waldenberger}). Singor et al. 
\cite{Singorealtri} consider a Heston-type inflation 
model in combination with a Hull-White model for 
interest rates, with non-zero correlations. Hughston and 
Macrina \cite{HughstonMacrina} propose a discrete 
time model based on utility functions. Haubric et al. 
\cite{HaubricPennacchiRitchken} develop a discrete time model of 
nominal and real bond yield curves based on several 
stochastic drivers. 

Our model does not fall within the class of affine models, 
and does not reduce to other known models. 
Therefore a certain amount of mathematical work is 
needed to study it. 
In particular we have to prove first of all that it is well posed, i.e. 
that there is one and only one triple of stochastic 
processes that satisfies the above description. 
We then derive the valuation equation for the price of 
a derivative with a continuous payoff with sublinear growth. 
The valuation equation on a time horizon $[0,T]$ 
reduces to a series of terminal value problems 
on the time intervals $[t_i,t_{i+1})$ (the time intervals on 
which inflation is constant) for an equation that can 
be seen as a simple parabolic 
Partial Integro-Differential Equation. The equation is the 
same for all intervals, but the terminal value is 
different on each interval and is defined by a backward 
recursion: On the interval $[t_i,t_{i+1})$ the terminal value depends on 
the solution on the interval $[t_{i+1},t_{i+2})$. The equation is 
not uniformly parabolic because 
the second order coefficient is not bounded away from 
zero, therefore it is not obvious that a classical solution 
exists. However we are able to prove existence and 
uniqueness of the viscosity solution by applying a result of Costantini, Papi 
and D'Ippoliti \cite{CostantiniPapiD'Ippoliti2012}. 

Although in some special cases the price might admit a closed 
form, or might be approximated by a closed form, in 
general it has to be computed numerically. This can be 
done very easily and efficiently by the  
semi-implicit, second order, finite difference scheme that 
we propose in Section \ref{sectionTheImplementation}, 
modifying the scheme proposed in Zhu and Li 
\cite{ZhuLi2003}. One reason the scheme is quite 
accurate is that it does not impose an artificial 
boundary condition at the boundary where the short 
term interest rate is zero. 

The paper is organized as follows. In Section~\ref{sectionTheModel},
we introduce the mathematical model and check that it is well posed. 
InSection~\ref{sectionVE} we derive the valuation
equation and prove that the price is the only solution 
in the viscosity sense. 
In Section~\ref{sectionTheImplementation} we describe 
the numerical scheme and compute numerically the price 
in an example. 


\section{The model}\label{sectionTheModel}
\begin{figure}[htp]
\begin{center}
    \includegraphics[scale=0.7]{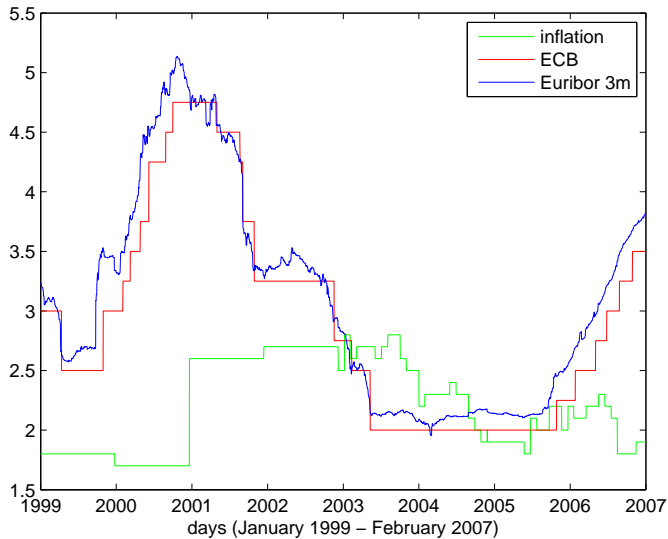}
    \caption[]{{\small{\em Evolution of European Inflation, ECB and short-term interest rates in the period
        January 1999- February 2007 (daily data).}}}
\label{figmodel}
\end{center}
\end{figure}

Figure~\ref{figmodel} plots the ECB interest rate together with
the inflation rate and the short-term interest rate. 
The primary goal of the ECB is to maintain price stability,
i.e., to keep inflation within a desired range (close to 2\%). The
inflation target is achieved through periodic adjustments of the ECB
official interest rate and, consequently, of short-term
interest rates. This behaviour is the core of this
work: In this section we formulate a dynamical model  
that describes the relationship among inflation, 
the ECB and the short-term interest rates, under a 
martingale measure. 

From now on, we consider $t\in [0,T]$ with $T<+\infty$, and we
fix the probability space $\left(\Omega ,\mathcal{F},\mathbb{P}\right
)$.


\subsection{European Inflation}\label{subsecInflation}
The value of the European rate of inflation is officially made known
once a month, hence we model it as a stochastic process that jumps
at fixed times, with jump sizes depending on the previous value of
the inflation and on the spread between the official ECB interest
rate and the short-term
rate, and it is constant between two jumps. \\
Specifically, using the usual convention that one year is an
interval of length one, let $\mathcal{T}:=\{t_i\}_{i\geq 0, \ldots,N}$, where $t_i=i\Theta,$
(with $\Theta=\frac{1}{12}$) be the sequence of times at which the values of
the inflation process, $\left\{\Pi(t_{i})\right\}_{t_i \in \mathcal{T}}$, are
observed.\\
The evolution is then given by
\begin{equation}\label{dynamicsInfl}
        \left\{\begin{array}{ll}
            \Pi(0)=\Pi_{0},\\
            \Pi(t)=\Pi(t_{i}),& \quad t_i\leq t < t_{i+1},\\
            \Pi(t_{i+1})=\gamma\left(\Pi(t_{i}),R(t_{i+1}^-),R^{sh}(t_{i+1}^-)\right)+\epsilon_{i+1},
                & \quad t=t_{i+1},
                \end{array}
        \right.
\end{equation}
where $\gamma$ is a linear function defined by
\begin{equation}\label{gamma}
    \gamma(\pi,r,z)=\alpha\pi+k^{\Pi}(\pi^*-\pi)+\beta
            \left(r-z\right),       
\end{equation}
with $\alpha,~\beta \in \mathbb{R}$ and $k^{\Pi},~\pi^* \in
\mathbb{R}_+$ constant parameters such that $0<\alpha-k^{\Pi}<1$.\\
The fluctuations $\{\epsilon_{i}\}_{i=0, \ldots,M}$ are i.i.d. random variables
distributed according to the $\mathcal{N}(0,v^{2})$ law, and
$R(t)$ and $R^{sh}(t)$, for ${t\in [0,T]}$, are the
interest rate processes which will be introduced in
Sections~\ref{subsecECB} and \ref{subsecShort}.\\
We can see that 
\begin{equation}\label{g}
\gamma(\pi,r,z)=\pi+\left(k^{\Pi}-\alpha+1\right)\left[\frac{k^{\Pi}\pi^*+\beta
            \left(r-z\right)}{k^{\Pi}-\alpha+1}-\pi)\right]
\end{equation}
and hence the condition $0<\alpha-k^{\Pi}<1$ yields that the process
$\Pi$  satisfies the mean-reversion property towards $\frac{k^{\Pi}\pi^*+\beta\left(r-z\right)}{k^{\Pi}-\alpha+1}$.


\subsection{European Central Bank Interest Rate}\label{subsecECB}

Looking at figure \ref{figmodel},
we can see some important facts about the ECB interest rate.
The level of the rate is persistent, hence 
the sample path is a step function; The changes 
are multiples of 25 basis points (bp); A change is 
often followed by additional changes, frequently in the 
same direction  (especially in the last years). 
Therefore we can model the ECB 
interest rate, $R(t)$, as a continuous time, 
pure jump process with finitely many possible upward and downward jump values. 
These jumps occur at random times $\{\vartheta_i\}$, and their size 
is equal to $k\delta$ with $\delta =0.0025$ and $k\in \{-m,...,-1
,1,...,m\}$.
When the level of
the official interest rate is low, there is a tendency to avoid
further downward jumps, or, at least, to reduce the occurrence of
this type of jumps. To take into account this effect, we suppose that the jump
intensity is a function $\lambda=\lambda(\pi,r)$ of the current values of the inflation
process and of the ECB interest rate and the probability of occurence of 
a jump $k\delta$ also depends upon the current values of inflation and the ECB 
interest rate, i.e. it is a function $p(\pi ,r,k\delta )$.\\
Since, by definition, an interest rate is always larger than -1, we can assume,without loss 
of generality, $R(t)>\underline r, \quad \underline r\geq -1$ , for all $t>0$.
In addiction, we suppose that there exists a maximum value $0<\overline r<+\infty$ such that 
$R(t)<\overline r,$ for all $t>0.$
Consistently we assume that 
\begin{equation}\label{roverandunderline}p(\pi ,r,k\delta )=0\qquad\mbox{\rm for }
r+k\delta\notin (\underline r,\overline r),\end{equation}
and that 
\begin{equation}\label{overlineLambda}
    \overline{\lambda}:=\sup_{(\pi,r) \in
    \mathbb{R}\times (\underline{r},\overline{r})}\lambda(\pi,r))<+\infty,
\end{equation}
Of course we suppose that
\[\sum_{k\in \{-m,...,-1,1,...,m\}} p(\pi ,r,k\delta )=1 \qquad \forall \pi\in \mathbb{R},\, r\in [\underline r,\overline r]\]
Moreover in view of Section \ref{sectionVE}, we will make the following additional 
assumptions on 
$p(\cdot ,\cdot ,k\delta )$ and $\lambda =\lambda (\cdot ,\cdot )$: For $
k=1,...,m$, 
\begin{equation}\label{continuitapelambda}p(\cdot ,\cdot ,k\delta 
)\mbox{\rm \ and }\lambda =\lambda (\cdot ,\cdot )\mbox{\rm \ are continuous on }\mathbb{
R}\times [\underline r,\overline r]\end{equation}
and
\begin{equation}\label{derivatefinitepelambda}\begin{array}{c}
p(\pi ,\cdot ,k\delta )\mbox{\rm \ has a finite left derivative at }\overline 
r-k\delta\\
p(\pi ,\cdot ,-k\delta )\mbox{\rm \ has a finite right derivative at }\underline 
r+k\delta\end{array}
\end{equation}
The value $\overline {\lambda}$
can be considered as the intensity of a Poisson process $N=N(t)$ and, 
using this fact, henceforth, we will consider the ECB interest rate as the solution 
of the following stochastic equation 
\begin{equation}\label{dynamicsECB}R(t)=R_0+\int_0^tJ\left(\Pi (s^{
-}),R(s^{-}),U_{N(s^{-})+1}\right)~dN(s),\end{equation}
where $\{U_n\}_{n \geq 0}$ are i.i.d. $[0,1]$-uniform random
variables, independent of $N$,
\begin{equation}\label{J}\begin{array}{lll}
J(\pi ,r,u)&:=&-m\delta\mathbf{1}_{(0,1]}(q(\pi ,r,-m\delta ))
\mathbf{1}_{[0,q(\pi ,r,-m\delta )]}(u)\\
&&+\sum_{k=-m+1}^mk\delta~\mathbf{1}_{(0,1]}(q(\pi ,r,k\delta ))\mathbf{
1}_{(\sum_{h=-m}^{k-1}q(\pi ,r,h\delta ),\sum_{h=-m}^kq(\pi ,r,h\delta 
)]}(u),\end{array}
\end{equation}
for $u\in [0,1]$, and
\begin{equation}\label{probabilities}
    \left\{
    \begin{array}{ll}
     q(\pi,r,k\delta):=\frac{p(\pi,r,k\lambda)\lambda(\pi,r)}{\overline{\lambda}},\qquad \quad \mbox{\rm if } k\neq 0\\
        q(\pi,r,k\delta):=1-\frac{\lambda(\pi,r)}{\overline{\lambda}}\qquad \qquad \mbox{\rm if } k=0.\\
    \end{array}
    \right.
\end{equation}


\subsection{Short-term Interest Rate}\label{subsecShort} 

Consistently with empirical
observations, we model the evolution of the short-term interest rate, $
R^{sh}$, 
as a mean-reverting Ito process with coefficients 
depending on the ECB interest rate, $R$, and hence  
indirectly on the inflation $\Pi$ as well. 
More precisely, we suppose that $R^{sh}$ satisfies 
the following equation
\begin{equation}\label{dynamicsShort}\left\{\begin{array}{ll}
dR^{sh}(t)=k^{sh}\left(b(R(t))-R^{sh}(t)\right)dt+\overline {\sigma}\left
(|R(t)-R^{sh}(t)|^2\right)\sqrt {|R^{sh}(t)|}dW(t),\\
R^{sh}(0)=R^{sh}_0,\end{array}
\right.\end{equation}
where $k^{sh} \in \mathbb{R}_+$ is a constant parameter and
$\{W_{t}\}_{t \in [0,T]}$ is a standard Wiener process. The function
$b(r)$ is defined by
\begin{equation}\label{b}b(r)=b_0+b_1r,\end{equation}
where $b_0,b_1\in\mathbb{R}$ are constant parameters and 
$\inf_{(\underline r,\overline r)}b(r)>0.$
The volatility coefficient $\overline {\sigma}$ is allowed to depend on the spread 
between $R^{sh}$ and $R$, so as to model the fact that higer values of the spread 
may lead to higher volatility of $R^{sh}$.
Moreover $\overline {\sigma}$ is nonnegative, 
and, in view 
of Section \ref{sectionVE}, the square of $\overline {\sigma}$, $\overline {
\sigma}^2$, satisfies the 
following assumptions:
\begin{equation}\label{sigmabarproperty1}\overline {\sigma}^2\in 
{\cal C}^2([0,\infty )),\end{equation}
\begin{equation}\label{sigmabarproperty2}\overline {\sigma}^2(q)\leq
\sigma_1(1+\sqrt q)\\
,\,\,\,q\in [0,+\infty )\\
.\end{equation}
For instance one can take 
\[\overline {\sigma}(q)=\sigma_0\frac q{1+q}.\]
Moreover, in analogy to the CIR model, we assume that
\begin{equation}\label{bandsigmaproperty}k^{sh}\inf_{(\underline 
r,\overline r)}b(r)\geq\frac 12\overline {\sigma}^2\left(\overline 
r-(\underline r\wedge 0)\right).\end{equation}


\subsection{Well-posedness of the model}\label{subsecwellpos}
We suppose that the sources of randomness in (\ref{dynamicsInfl}),
(\ref{dynamicsECB}) and (\ref{dynamicsShort}), that is
$\{\epsilon_i\}_{1\leq i\leq M}$, $\{N(t)\}_{t\geq 0}$, $\{U_n\}_{
n\geq 1}$ and $\{W(t)\}_{t\geq 0}$, are mutually
independent. All information is given by the following filtration
\begin{equation}\label{filtrationbig}\mathcal{F}_t:=\sigma\left(\left
\{\Pi_0,R_0,R^{sh}_0,\epsilon_{I(s)},N(s),W(s),U_{N(s)},s\leq t\right
\}\right),\end{equation}
where $I(s)$ is the number of jumps of inflation up to time $s$,
namely,
\begin{equation}\label{I}
    I(s):=\max\{i \geq 0:t_i \leq s\}, s\geq 0.
\end{equation}
The model introduced in Sections
\ref{subsecInflation}-\ref{subsecShort} is well posed, in the sense
that there exists one and only one stochastic process
$\left(\Pi,R,R^{sh}\right)$ verifying
(\ref{dynamicsInfl}), (\ref{dynamicsECB}) and (\ref{dynamicsShort})
in a strong sense, as we prove in the following theorem.
\begin{teo}\label{theowellpos}
For every triple of $\mathbb{R}\times (\underline r,\overline r)\times 
(0,+\infty )$-valued r.v.'s $\left(\Pi_0,R_0,R^{sh}_0\right)$, 
there exists one and only one stochastic process
    $(\Pi,R,R^{sh})$ defined on $(\Omega ,\mathcal{F},\mathbb{
P})$, $\{\mathcal{F}_t\}$-adapted, 
such that $(\ref{dynamicsInfl})$, $(\ref{dynamicsECB})$ and 
$(\ref{dynamicsShort})$ 
are $\mathbb{P}$-a.s. verified. 
It holds $R^{sh}(t)>0$ for all $t\geq 0$, almost surely.
\end{teo}
\begin{proof}
Setting $t_0:=0$ and $(\Pi (0),R(0),R^{sh}(0)):=\left(\Pi_0,R_0,R^{
sh}_0\right)$, we 
claim that, given a triple 
of $\mathbb{R}\times [\underline r,\overline r]\times\R$-valued, $
\{\mathcal{F}_{t_i}\}$-measurable r.v.'s 
$\left(\Pi (t_i),R(t_i),R^{sh}(t_i)\right)$, $(\Pi ,R,R^{sh})$ is pathwise uniquely 
defined on the interval $[t_i,t_{i+1}]$ and 
$\left(\Pi (t_{i+1}),R(t_{i+1}),R^{sh}(t_{i+1})\right)$ is $\{\mathcal{
F}_{t_{i+1}}\}$-measurable. 
To see this, observe, first of all, the probability that $N$ jumps at any of $
t_1,...,t_M$ is zero. 
Therefore, denoting by $\{\vartheta_n\}$ the jump times of 
$N$, $R$ can be defined simply in the following way: 
If $N(t_{i+1})>N(t_i)$, 
\begin{eqnarray*}
R(\vartheta_{N(t_i)+n}):=R(\vartheta_{N(t_i)+n-1}\vee t_i)+J\left(\Pi (t_
i),R(\vartheta_{N(t_i)+n-1}\vee t_i),U_{N(t_i)+n}\right),\\
\mbox{\rm \ for } 1\leq n
\leq N(t_{i+1})-N(t_i),
\end{eqnarray*}
\[R(t):=R(\vartheta_{N(t_i)+n-1}\vee t_i)\quad\mbox{\rm \ for }\vartheta_{N(t_i)
+n-1}\vee t_i\leq t<\vartheta_{N(t_i)+n},\quad 1\leq n\leq N(t_{i+1})-
N(t_i),\]
\[R(t):=R(\vartheta_{N(t_{i+1})})\qquad \mbox{\rm \ for }\vartheta_{N(t_{i+1})}\leq 
t\leq t_{i+1}.\]
If $N(t_{i+1})=N(t_i)$,  
\[R(t):=R(t_i),\qquad \mbox{\rm \ for }t_i\leq t\leq t_{i+1}.\]
Note that $R(\vartheta_{N(t_i)+n}\wedge t_{i+1})$ is 
$\{\mathcal{F}_{\vartheta_{N(t_i)+n}\wedge t_{i+1}}\}$-measurable for all $
n\geq 1$, 
and that $R(t_{i+1})$ is $\{\mathcal{F}_{t_{i+1}}\}$-measurable.
Denote $\vartheta^i_0:=t_i$, $\vartheta^i_n:=\vartheta_{N(t_i)+n}\wedge t_{i+1}$, $
n\geq 1$. In each subinterval 
$[\vartheta^i_{n-1},\vartheta^i_n]$, we can write the equation 
$(\ref{dynamicsShort})$ as 
\begin{eqnarray}
R^{sh}(\vartheta^i_{n-1}+t)=R^{sh}(\vartheta^i_{n-1})+\int_0^tk^{sh}\left(
b(R(\vartheta^i_{n-1}))-R^{sh}(\vartheta^i_{n-1}+s)\right)ds\\
+\int_0^t\overline {\sigma}\left(|R(\vartheta^i_{n-1})-R^{sh}(\vartheta^i_{
n-1}+s)|^2\right)\sqrt {|R^{sh}(\vartheta^i_{n-1}+s)|}dW^{\vartheta^i_{n-1}}
(s)\\ 
0\leq t\leq\vartheta^i_n-\vartheta^i_{n-1},
\label{IkedaWatanabe}\end{eqnarray}
where $W^{\vartheta^i_{n-1}}(s):=W(\vartheta^i_{n-1}+s)-W(\vartheta^i_{n-1})$. 
Since $W$ is independent of $N$, $W^{\vartheta^i_{n-1}}$ is a standard 
Brownian motion, independent of $\vartheta^i_n-\vartheta^i_{n-1}$. Moreover, if 
$R^{sh}(\vartheta^i_{n-1})$ is $\{\mathcal{F}_{\vartheta^i_{n-1}}\}$-measurable (hence
$\big(R(\vartheta^i_{n-1}),R^{sh}(\vartheta^i_{n-1})\big)$ is $\{\mathcal{F}_{
\vartheta^i_{n-1}}\}$-measurable) $W^{\vartheta^i_{n-1}}$ 
is independent of $\big(R(\vartheta^i_{n-1}),R^{sh}(\vartheta^i_{n-1})\big)$. 
The diffusion coefficient in 
$(\ref{IkedaWatanabe})$ is locally Holder continuous by  
$(\ref{sigmabarproperty1})$, and has sublinear growth by 
$(\ref{sigmabarproperty1})$. Therefore, by the Corollary 
to Theorem 3.2, Chapter 4, of \cite{IkedaWatanabe}, there exists one and only one 
strong solution to $(\ref{IkedaWatanabe})$ 
(the Corollary to Theorem 3.2 of \cite{IkedaWatanabe} assumes global Holder 
continuity, but, as pointed out in the comment 
immediately preceding Theorem 3.2, its statement can be 
localized and it yields existence and uniqueness of the 
strong solution up to the explosion time; $(\ref{sigmabarproperty1})$ and 
Theorem 2.4, Chapter 4, of \cite{IkedaWatanabe} ensure that the 
explosion time is infinite). 
Then $R^{sh}(\vartheta^i_{n-1}+t)$ is pathwise uniquely defined for 
$0\leq t\leq\vartheta^i_n-\vartheta^i_{n-1}$ and $R^{sh}(\vartheta^i_n)$ is $\{\mathcal{
F}_{\vartheta^i_n}\}$-measurable. 
Since $R^{sh}(\vartheta^i_0)=R^{sh}(t_i)$ is $\{\mathcal{F}_{t_i}\}$-measurable, i.e. 
$\{\mathcal{F}_{\vartheta^i_0}\}$-measurable, we see, by induction, that 
$R^{sh}$ is pathwise uniquely defined on $[t_i,t_{i+1}]$ and 
$R^{sh}(t_{i+1})$ is $\{\mathcal{F}_{t_{i+1}}\}$-measurable. 
By setting  
\[\Pi (t_{i+1})=\gamma\left(\Pi (t_i),R(t_{i+1}),R^{sh}(t_{i+1})\right
)+\epsilon_{i+1},\]
our claim is proved. 
Finally, let us show that $R^{sh}(t)>0$ for all $t\geq 0$, almost 
surely. Let $\alpha_n:=\inf\{t\geq 0:\,R^{sh}(t)\leq\frac 1n\}$. Then it will be 
enough to show, for every $z_0>0$, for $R^{sh}_0=z_0$, that
\[\P (\alpha_n\leq t)\rightarrow_{n\rightarrow\infty}0,\qquad\forall 
t>0.\]
To see this, consider, for $z>0$, the function 
\[V_1(z):=z^2-\ln(z).\]
We have 
\[\begin{array}{llccc}
&k^{sh}\left(b(r)-z\right)V_1(z)'+\frac 12\overline {\sigma}^2\left
(|r-z|^2\right)V_1(z)^{\prime\prime}\\
&=k^{sh}\left(b(r)-z\right)\big(2z-\frac 1z\big)+\frac 12\overline {
\sigma}^2\left(|r-z|^2\right)\big(2z+\frac 1z\big)\\
&=\1_{\{z\leq\overline r\}}\frac 1z\bigg(\frac 12\overline {\sigma}^
2\left(|r-z|^2\right)-k^{sh}b(r)\bigg)+\1_{\{z>\overline r\}}\frac 
1z\bigg(\frac 12\overline {\sigma}^2\left(|r-z|^2\right)-k^{sh}b(
r)\bigg)\\
&\qquad\qquad\qquad\qquad +k^{sh}b(r)+2z\bigg(\frac 12\overline {
\sigma}^2\left(|r-z|^2\right)+k^{sh}b(r)-k^{sh}z\bigg)\\
&\leq\frac 1{2\overline r}\overline {\sigma}^2\left(|r-z|^2\right
)+k^{sh}b(r)+2z\bigg(\frac 12\overline {\sigma}^2\left(|r-z|^2\right
)+k^{sh}b(r)\bigg)\\
&\leq C(1+z^2),\end{array}
\]
where the last but one inequality follows from 
$(\ref{bandsigmaproperty})$ and $z>0$, 
and the last one follows from 
$(\ref{sigmabarproperty2})$. 
Let $\beta_k:=\inf\{t\geq 0:\,R^{sh}(t)\geq k\}$. By applying Ito's 
formula and taking expectations, we obtain 
\begin{eqnarray*}
&&\E[V_1(R^{sh}(t\wedge\alpha_n\wedge\beta_k)]\\
&\leq&V_1(z_0)+C\E\bigg[\int_0^{t\wedge\alpha_n\wedge\beta_k}\big
(1+R^{sh}(s)^2\big)ds\\
&\leq&V_1(z_0)+C\int_0^t\bigg(1+\E\bigg[V_1(R^{sh}(s\wedge\alpha_
n\wedge\beta_k)\bigg]\bigg)ds,\end{eqnarray*}
which implies, by Gronwall's Lemma and by taking limits 
as $k\rightarrow\infty$,  
\[\E[V_1(R^{sh}(t\wedge\alpha_n)]\leq\big(V_1(z_0)+Ct\big)e^{Ct},\]
and hence, 
\[\ln(n)\P (\alpha_n\leq t)\leq\big(V_1(z_0)+Ct\big)e^{Ct}.\]
\end{proof}


\section{The Valuation Equation}\label{sectionVE}

In this section, we derive the valuation equation for 
a contingent claim with maturity $T$ and
payoff $\Phi (\Pi (T),R(T),R^{sh}(T))$. 
As it is well known, under a risk neutral measure, 
the price $P(t)$ of such a contingent claim can
be expressed as the expected discounted payoff, namely, 
 
\[P(t)=\E\bigg[\exp\bigg(-\int_t^TR^{sh}(s)ds\bigg)\Phi (\Pi (T),
R(T),R^{sh}(T))\bigg|{\cal F}_t\bigg].\]
Therefore, due to the Markov property of $(\Pi ,R,R^{sh})$, 
\[P(t)=\varphi (t,\Pi (t),R(t),R^{sh}(t)),\]
where 
\begin{eqnarray}&&\!\!\!\!\!\!\!\!\varphi (t,\pi ,r,z):=\nonumber\\
&&\!\!\!\!\!\!\!\!\!\!\!\!\E\bigg[\exp\bigg(-\int_t^T
R^{sh}(s)ds\bigg)\Phi (\Pi (T),R(T),R^{sh}(T))\bigg|(\Pi (t),R(t)
,R^{sh}(t))=(\pi ,r,z)\bigg].\qquad\label{price}\end{eqnarray}

We assume that the payoff is continuos and satisfies 
\begin{equation}|\Phi (\pi ,r,z)|\leq C_0(1+|\pi |+z),\quad\pi\in\R
,\,r\in\underline r,\overline r),\,\,\,z>0,\label{payoff}\end{equation}
for some positive constant $C_0$. 
Supposing, for simplicity, that $T=t_M$ (the last time 
inflation is measured), we are going to show that 
\begin{equation}\varphi (t,\pi ,r,z)=\varphi^i(t,\pi 
,r,z),\quad t_i\leq t<t_{i+1},~i=0,\ldots ,M-1.\label{varphi1}\end{equation}
where, for each value of inflation, $\pi ,$ $\varphi^i(t_i+\cdot 
,\pi ,\cdot ,\cdot )$ is the unique 
solution of the terminal value problem on 
$[0,\Theta ]$ for an equation that can be viewed as a simple 
parabolic Partial 
Integro-Differential Equation (with coefficients 
depending on the parameter $\pi$). The equation is the same 
for all $i$'s, but the terminal value  
is different for each $i$ and is defined recursively from 
$\varphi^{i+1}(t_{i+1},\cdot ,\cdot ,\cdot )$, for $i=0,...,M-2$, and from $
\Phi$ for $i=M-1$ 
(see Lemma \ref{recursion} below). 
As it will be seen in the next section, this series of 
parametric terminal value problems can be easily solved 
numerically. 

To carry out our program, we need to introduce the 
parametrized process $(R^{\pi},R^{sh,\pi})$ defined by 
the following system of stochastic equations 
\begin{equation}\left\{\begin{array}{cc}
R^{\pi}(t)=R^{\pi}_0+\int_0^tJ(\pi ,R^{\pi}(s^{-}),U_{N(s^{-})+1}
)dN(s), \label{param}\\
R^{sh,\pi}(t)=R^{sh,\pi}_0+\int_0^tk^{sh}\left(b(R^{\pi}(s))-R^{s
h,\pi}(s)\right)ds\\
+\int_0^t\overline {\sigma}\left(|R^{\pi}(s)-R^{
sh,\pi}(s)|^2\right)\sqrt {|R^{sh,\pi}(s)|}dW(s).\end{array}\right.
\end{equation}

\begin{pro}\label{param-wdef}

Let $(N,\{U_n\},W)$  be as in Theorem \ref{theowellpos}. 
For each $\pi\in\R$, for each $(\underline r,\overline r)\times (
0,\infty )$-valued random 
variable $(R^{\pi}_0,R^{sh,\pi}_0)$, independent of $(N,\{U_n\},W
)$, there exists 
a unique strong solution of the system of 
equations $(\ref{param})$. 

It holds $R^{\pi}(t)\in (\underline r,\overline r)$, $R^{sh,\pi}(
t)>0$ for all $t\geq 0.$ 

The solution corresponding to 
$(R^{\pi}_0,R^{sh,\pi}_0)=(r,z)$ will be denoted by $(R^{\pi}_r,R^{
sh,\pi}_{(r,z)})$. 

\end{pro}

\begin{proof}
The proof is analogous to the proof of Theorem 
\ref{theowellpos}.
\end{proof}

\begin{lemma}\label{param-prop}
For any continuous function 
$f:\R\times (\underline r,\overline r)\times (0,\infty )\rightarrow\R$  such that $
|f(\pi ,r,z)|\leq C(1+|\pi |+z)$ for 
$(\pi ,r,z)\in\R\times (\underline r,\overline r)\times (0,\infty 
)$, 
$\E\bigg[e^{-\int_0^{\Theta}R^{sh,\pi}_{(r,z)}(s)ds}\big|f\big(\pi 
,R^{\pi}_r(\Theta ),R^{sh,\pi}_{(r,z)}(\Theta )\big)\big|\bigg]$ is finite 
for every $(\pi ,r,z)$ and the function 
\[F(t,\pi ,r,z):=\E\left[e^{-\int_0^{\Theta -t}R^{sh,\pi}_{(r,z)}
(s)ds}f\left(\pi ,R^{\pi}_r(\Theta -t),R^{sh,\pi}_{(r,z)}(\Theta 
-t)\right)\right]\]
is continuous on $[0,\Theta ]\times\R\times (\underline r,\overline 
r)\times (0,\infty )$ and satisfies 
\[|F(t,\pi ,r,z)|\leq C'(1+|\pi |+z),\qquad (t,\pi ,r,z)\in [0,\Theta 
]\times\R\times (\underline r,\overline r)\times (0,\infty ).\]
\end{lemma}

\begin{proof}
Let us show preliminarly that, for every $T>0$, 
\begin{equation}\E[R^{sh,\pi}_{(r,z)}(t)]\leq C_T(1+z),\qquad 0\leq 
t\leq T,\,\,\,\,\pi\in\R,\,\,\,r\in (\underline r,\overline r),\,\,\,
z>0,\label{mom1}\end{equation}
and, for every $p\geq 2$, 
\begin{equation}\E[R^{sh,\pi}_{(r,z)}(t)^p]\leq C_T(1+z^p),\qquad 
0\leq t\leq T,\,\,\,\,\pi\in\R,\,\,\,r\in (\underline r,\overline 
r),\,\,\,z>0.\label{momp}\end{equation}
In order to prove $(\ref{momp})$, consider a sequence of bounded, nonnegative $
{\cal C}^2$ 
functions $\{f_n\}$ such that $f_n(z)=z^p$ for 
$0<z\leq n$ and $f_n(z)\leq z^p$ for all $z>0$, and let $\alpha_n
:=\inf\{t\geq 0:\,R^{sh,\pi}_{(r,z)}(t)\geq n\}$. 
By applying Ito's Lemma to the semimartingale $R^{sh,\pi}_{(r,z)}$, 
and to the function $f_n$, and taking expectations, we 
obtain 
\begin{eqnarray*}
&&\E[R^{sh,\pi}_{(r,z)}(t\wedge\alpha_n)^p]\\
&=&\E[f_n(R^{sh,\pi}_{(r,z)}(t\wedge\alpha_n))]\\
&=&f_n(z)+k^{sh}\E\bigg[\int_0^{t\wedge\alpha_n}\bigg([(R^{\pi}_
r(s))-R^{sh,\pi}_{(r,z)}(s)]R^{sh,\pi}_{(r,z)}(s)^{p-1})\\
&&\qquad\qquad\qquad\qquad\qquad\qquad\qquad +\frac {p(p-1)}2\overline {
\sigma}^2(|R^{\pi}(s)-R^{sh,\pi}(s)|^2)R^{sh,\pi}(s)^{p-1}\bigg)ds\bigg
]\\
&\leq&z^p+C\E\bigg[\int_0^{t\wedge\alpha_n}\bigg(1+R^{sh,\pi}_{(r
,z)}(s))^p\bigg)ds\bigg]\\
&\leq&z^p+C\int_0^t\bigg(1+\E\bigg[R^{sh,\pi}_{(r,z)}(s\wedge\alpha_
n))^p\bigg]\bigg)ds,\end{eqnarray*}
where the last but one inequality follows from 
$(\ref{sigmabarproperty2})$. 

\noindent Therefore, by Gronwall's Lemma and Fatou's Lemma, 
\[\E[R^{sh,\pi}_{(r,z)}(t)^p]\leq\bigg(z^p+CT\bigg)e^{CT}.\]
$(\ref{mom1})$ can be proved in an analogous manner. 
\noindent $(\ref{mom1})$ yields both that 
$$\E\bigg[e^{-\int_0^{\Theta}R^{sh,\pi}_{(r,z)}(s)ds}\big|f\big(\pi 
,R^{\pi}_r(\Theta ),R^{sh,\pi}_{(r,z)}(\Theta )\big)\big|\bigg]\;\;\mbox{\rm is finite, and}\;\;\;\;
|F(t,\pi ,r,z)|\leq C'(1+|\pi |+z),$$ for 
$(t,\pi ,r,z)\in [0,\Theta ]\times\R\times (\underline r,\overline 
r)\times (0,\infty )$. 
We are left with proving continuity of $F$. Let 
$(t_n,\pi_n,r_n,z_n)\rightarrow (t,\pi ,r,z)$. Then $\{R^{\pi}_{r_
n}\}$ converges to $R^{\pi}_r$ uniformly 
over compact time intervals, almost surely. In addition 
$\{R^{sh,\pi}_{(r_n,z_n)}\}$ is relatively compact by Theorems 3.8.6 and 
3.8.7 of \cite{EthierKurtz}, the Burkholder-Davies-Gundy 
inequality and $(\ref{momp})$ with $p=4$, and every limit point is 
continuous by Theorem 3.10.2 of \cite{EthierKurtz}. Therefore $\{(R^{\pi}_{
r_n},R^{sh,\pi}_{(r_n,z_n)})\}$ is 
relatively compact. By Theorem 2.7 of 
\cite{KurtzProtter}, every limit point of $\{(R^{\pi}_{r_n},R^{sh,\pi}_{(
r_n,z_n)})\}$ 
satisfies $(\ref{param})$ with $(R^{\pi}_0,R^{sh,\pi}_0)=(r,z)$. 
Since the solution to $(\ref{param})$ is (strongly and 
hence weakly) unique, we can conclude that 
$\{(R^{\pi}_{r_n},R^{sh,\pi}_{(r_n,z_n)})\}$ converges weakly to $
(R^{\pi}_r,R^{sh,\pi}_{(r,z)})$. 
The assertion then follows by observing that $(\ref{momp})$ 
implies that the random variables 
$\bigg\{e^{-\int_0^{\Theta -t_n}R^{sh,\pi}_{(r_n,z_n)}(s)ds}f\left
(\pi_n,R^{\pi}_{r_n}(\Theta -t_n),R^{sh,\pi}_{(r_n,z_n)}(\Theta -
t_n)\right)\bigg\}$ 
are uniformly integrable. 
\end{proof}

\noindent
For a continuous function $f:\R\times (\underline r,\overline r)\times 
(0,\infty )\rightarrow\R$  such that 
$|f(\pi ,r,z)|\leq C(1+|\pi |+z)$ for $(\pi ,r,z)\in\R\times (\underline 
r,\overline r)\times (0,\infty )$, set 
\begin{equation}Bf(\pi ,r,z):=\frac 1{\sqrt {2\pi}v}\int_{\mathbb{
R}}f\left(\gamma (\pi ,r,z)+u,r,z\right)\exp\left(-\frac {u^2}{2v^
2}\right)du.\label{B}\end{equation}

\begin{lemma}\label{B-prop}
For any continuous function 
$f:\R\times (\underline r,\overline r)\times (0,\infty )\rightarrow\R$  such that $
|f(\pi ,r,z)|\leq C(1+|\pi |+z)$ for 
$(\pi ,r,z)\in\R\times (\underline r,\overline r)\times (0,\infty 
)$, $Bf$ is continuous and satisfies 
\[|Bf(\pi ,r,z)|\leq C'(1+|\pi |+z),\qquad (\pi ,r,z)\in\R\times 
(\underline r,\overline r)\times (0,\infty ).\]
\end{lemma}

\begin{proof}\ 
For $(\pi_n,r_n,z_n)\rightarrow (\pi ,r,z)$ 
\[|f\left(\gamma (\pi_n,r_n,z_n)+u,r_n,z_n\right)|\leq C(\pi ,r,z
)(1+|u|)\]
and the first assertion follows by dominated 
convergence. 
In addition, by $(\ref{gamma})$, 
\[\begin{array}{rcl}
|Bf(\pi ,r,z)|&\leq&C\,\frac 1{\sqrt {2\pi}v}\int_{\mathbb{R}}\left
(1+|\gamma (\pi ,r,z)+u|+z\right)\exp\left(-\frac {u^2}{2v^2}\right
)du\\
&\leq&C\left(1+|\gamma (\pi ,r,z)|+v+z\right)\leq C_2C(1+|\pi |+z
).\end{array}
)\]
\end{proof}

\begin{lemma}\label{recursion}
Let $\varphi$ be the function defined by $(\ref{price})$. Then 
\[\varphi (t,\pi ,r,z)=\left\{\begin{array}{lll}
\Phi (\pi ,r,z),&t=t_M,\\
\varphi^i(t,\pi ,r,z),&t_i\leq t<t_{i+1},&i=0,...,M-1,\end{array}
\right.\]
where the functions $\varphi^i$ are defined recursively in the 
following way: 
\[\varphi^{M-1}(t_{M-1}+t,\pi ,r,z)=\E\left[e^{-\int_0^{\Theta -t}
R^{sh,\pi}_{(r,z)}(s)ds}B\Phi\left(\pi ,R^{\pi}_{r}(\Theta -t
),R^{sh,\pi}_{(r,z)}(\Theta -t)\right)\right],\]
and
\[\varphi^i(t_i+t,\pi ,r,z)=\E\left[e^{-\int_0^{\Theta -t}R^{sh,\pi}_{
(r,z)}(s)ds}B\left(\varphi^{i+1}(t_{i+1},\cdot ,\cdot ,\cdot )\right
)\left(\pi ,R^{\pi}_{r}(\Theta -t),R^{sh,\pi}_{(r,z)}(\Theta 
-t)\right)\right],\]
for $\, 0\leq t\leq\Theta\, $ and for $\, i=0,...,M-2$. 
\end{lemma}

\begin{proof}
For $t_{M-1}\leq t<t_M$, 
\begin{eqnarray*}
&&\E\bigg[\exp\bigg(-\int_t^{t_M}R^{sh}(s)ds\bigg)\Phi (\Pi (t_M)
,R(t_M),R^{sh}(t_M))\bigg|(\Pi (t),R(t),R^{sh}(t))\bigg]\\
&=&\E\bigg[\exp\bigg(-\int_t^{t_M}R^{sh}(s)ds\bigg)\Phi (\Pi (t_M
),R(t_M),R^{sh}(t_M))\bigg|{\cal F}_t\bigg]\\
&=&\E\bigg[\E\bigg[\exp\bigg(-\int_t^{t_M}R^{sh}(s)ds\bigg)\Phi (
\Pi (t_M),R(t_M),R^{sh}(t_M))\bigg|{\cal F}_{t_M^{-}}\bigg]\bigg|
{\cal F}_t\bigg]\\
\ &=&\E\bigg[\exp\bigg(-\int_t^{t_M}R^{sh}(s)ds\bigg)B\Phi (\Pi (
t),R(t_M),R^{sh}(t_M))\bigg|{\cal F}_t\bigg]\\
&=&\E\bigg[\exp\bigg(-\int_t^{t_M}R^{sh}(s)ds\bigg)B\Phi (\Pi (t)
,R(t_M),R^{sh}(t_M))\bigg|(\Pi (t),R(t),R^{sh}(t))\bigg]\end{eqnarray*}
Therefore, for $0\leq t<\Theta$, 
\begin{eqnarray*}
&&\varphi (t_{M-1}+t,\pi ,r,z)\\
&=&\E\bigg[\exp\bigg(-\int_{t_{M-1}+t}^{t_M}R^{sh}(s)ds\bigg)B\Phi 
(\Pi (t_{M-1}+t),R(t_M),R^{sh}(t_M))\bigg|(\Pi ,R,R^{sh})(t_{M-1}
+t)=(\pi ,r,z)\bigg]\\
&=&\E\bigg[\exp\bigg(-\int_{t_{M-1}+t}^{t_M}R^{sh}(s)ds\bigg)B\Phi 
(\pi ,R^{\pi}(t_M),R^{sh,\pi}(t_M))\bigg|(R^{\pi},R^{sh,\pi})(t_{
M-1}+t)=(r,z)\bigg]\\
&=&\E\bigg[\exp\bigg(-\int_t^{\Theta}R^{sh}(s)ds\bigg)B\Phi (\pi 
,R^{\pi}(\Theta ),R^{sh,\pi}(\Theta ))\bigg|(R^{\pi},R^{sh,\pi})(
t)=(r,z)\bigg]\\
&=&\E\bigg[\exp\bigg(-\int_0^{\Theta -t}R^{sh}_{(r,z)}(s)ds\bigg)
B\Phi (\pi ,R^{\pi}_r(\Theta -t),R^{sh,\pi}_{(r,z)}(\Theta -t))\bigg
],\end{eqnarray*}
where the last two inequalities follow from the fact 
that $(R^{\pi},R^{sh,\pi})$ is time homogeneous. 
Assuming inductively that $\varphi (t_{i+1},\pi ,r,z)=\varphi^{i+
1}(t_{i+1},\pi ,r,z)$, 
we have, for $t_i\leq t<t_{i+1}$, 
\begin{eqnarray*}
&&\E\bigg[\exp\bigg(-\int_t^{t_M}R^{sh}(s)ds\bigg)\Phi (\Pi (t_M)
,R(t_M),R^{sh}(t_M))\bigg|(\Pi (t),R(t),R^{sh}(t))\bigg]\\
&=&\E\bigg[\E\bigg[\exp\bigg(-\int_t^{t_M}R^{sh}(s)ds\bigg)\Phi (
\Pi (t_M),R(t_M),R^{sh}(t_M))\bigg|{\cal F}_{t_{i+1}}\bigg]\bigg|
{\cal F}_t\bigg]\\
&=&\E\bigg[\exp\bigg(-\int_t^{t_{i+1}}R^{sh}(s)ds\bigg)\varphi^{i
+1}(t_{t+1},\Pi (t_{i+1}),R(t_{i+1}),R^{sh}(t_{i+1}))\bigg|{\cal F}_
t\bigg]\\
&=&\E\bigg[\exp\bigg(-\int_t^{t_{i+1}}R^{sh}(s)ds\bigg)B\left(\varphi^{
i+1}(t_{i+1},\cdot ,\cdot ,\cdot )\right)\left(\Pi (t_{i+1}),R(t_{
i+1}),R^{sh}(t_{i+1})\right)\bigg|(\Pi (t),R(t),R^{sh}(t))\bigg]\end{eqnarray*}
and hence, by a computation analogous to that for the 
interval $[t_{M-1},t_M)$, 
\[\varphi (t_i+t,\pi ,r,z)=\E\left[e^{-\int_0^{\Theta -t}R^{sh,\pi}_{
(r,z)}(s)ds}B\left(\varphi^{i+1}(t_{i+1},\cdot ,\cdot ,\cdot )\right
)\left(\pi ,R^{\pi}_{r}(\Theta -t),R^{sh,\pi}_{(r,z)}(\Theta 
-t)\right)\right]\]
for $0\leq t\leq\Theta .$
\end{proof}

For an $\R$-valued diffusion process $X$ with time 
independent coefficients $b$ 
and $\sigma$, we know, by the Feynman-Kac formula, 
under suitable assumptions, that the function 
\[\psi (t,x):=\E\bigg[\exp\bigg(-\int_t^{T}X(s)ds\bigg)\Psi (
X(T))\bigg| X(t)=x \bigg]\\
=\E\bigg[\exp\bigg(-\int_0^{T-t}X_x(s)ds\bigg)\Psi (X_x(T-t))\bigg]
\]
where $X_x$ is the process starting at $x$ and
$\psi (t,x)$ is of class ${\cal C}^{1,2}$ and satisfies 
\[\frac {\partial\psi}{\partial t}(t,x)+b(x)\frac {\partial\psi}{
\partial x}(t,x)+\frac 12\sigma^2(x)\frac {\partial^2\psi}{\partial 
x^2}-x\psi (t,x)=0,\qquad 0\leq t<T,~x\in\R\\
,\]
\[\psi (\Theta ,x)=\Psi (x),\qquad x\in\R.\]
By analogy, we consider, for each fixed $\pi$, 
for each function $\varphi^i(t_i+\cdot ,\pi ,\cdot ,\cdot )$ 
defined in Lemma \ref{recursion}, the 
following equation, which reflects the dynamics of $(R^{\pi}_{(r,
z)},R^{sh,\pi}_{(r,z)})$: 
\begin{equation}\begin{array}{l}
\frac {\partial\psi}{\partial t}(t,r,z)+k^{sh}\left(b(r)-z\right)\frac {
\partial\psi}{\partial z}(t,r,z)+\frac 12\overline {\sigma}^2\left
(|r-z|^2\right)z\frac {\partial^2\psi}{\partial z^2}(t,r,z)\\
\qquad\qquad\quad\quad +\lambda (\pi ,r)\sum_{k=-m}^m\left[\psi\left
(t,r+k\delta ,z\right)-\psi\left(t,r,z\right)\right]p(\pi ,r,k\delta 
)-z\psi (t,r,z)=0,\\
\qquad\qquad\qquad\qquad\qquad\qquad\qquad\qquad\qquad\qquad\qquad
\qquad\qquad\quad 0\leq t<\Theta ,~r\in (\underline r,\overline r
),\,\,\,z>0,\\
\psi (\Theta ,r,z)=\Psi^i(\pi ,r,z),\qquad\qquad\qquad\qquad\qquad
\qquad\qquad\qquad r\in (\underline r,\overline r),\,\,\,z>0.\end{array}
\label{valeq}\end{equation}
where 
\begin{equation}\Psi^{M-1}:=B\Phi ,\quad\Psi^i:=B\left(\varphi^{i
+1}(t_{i+1},\cdot ,\cdot ,\cdot )\right)\quad\mbox{\rm for }i=0,.
..,M-2.\label{terminal}\end{equation}
$(\ref{valeq})$ is a not a standard partial differential 
equation and it is not clear whether it admits a classical 
solution. Instead  
we look for a viscosity solution. The definition of 
viscosity solution, in the present set up, is recalled 
below for the convenience of the reader. 

\begin{defi}\label{viscsol}
A viscosity solution of $(\ref{valeq})$ is a continuos 
function $\psi$ defined on $[0,\Theta ]\times (\underline r,\overline 
r)\times (0,\infty )$ such that, for 
each $(t_0,r_0,z_0)\in [0,\Theta ]\times (\underline r,\overline 
r)\times (0,\infty )$, for each $f\in {\cal C}^{1,2}\bigg([0,\Theta 
]\times\big((\underline r,\overline r)\times (0,\infty )\big)\bigg
)$ 
such that 
$$\sup_{(t,r,z)\in [0,\Theta ]\times (\underline r,\overline r)\times 
(0,\infty )}\left(\psi (t,r,z)-f(t,r,z)\right)=\left(\psi (t_0,r_
0,z_0)-f(t_0,r_0,z_0)\right)=0$$
it holds 
\[\begin{array}{l}
\frac {\partial f}{\partial t}(t,r,z)+k^{sh}\left(b(r)-z\right)\frac {
\partial f}{\partial z}(t,r,z)+\frac 12\overline {\sigma}^2\left(
|r-z|^2\right)z\frac {\partial^2f}{\partial z^2}(t,r,z)\\
\qquad\qquad\quad\quad +\lambda (\pi ,r)\sum_{k=-m}^m\left[\psi\left
(t,r+k\delta ,z\right)-\psi\left(t,r,z\right)\right]p(\pi ,r,k\delta 
)-z\psi (t,r,z)\geq 0,\\
\qquad\qquad\qquad\qquad\qquad\qquad\qquad\qquad\qquad\qquad\qquad
\quad 0\leq t<\Theta ,~r\in (\underline r,\overline r
),\,\,\,z>0,\\
f(\Theta ,r,z)\leq\Psi (\pi ,r,z),\qquad\qquad\qquad\qquad\qquad\qquad
\qquad\qquad r\in (\underline r,\overline r),\,\,\,z>0,\end{array}
\]
and, for each $f\in {\cal C}^{1,2}\bigg([0,\Theta ]\times\big((\underline 
r,\overline r)\times (0,\infty )\big)\bigg)$ such that 
\[\inf_{(t,r,z)\in [0,\Theta ]\times (\underline r,\overline r)\times 
(0,\infty )}\left(\psi (t,r,z)-f(t,r,z)\right)=\left(\psi (t_0,r_
0,z_0)-f(t_0,r_0,z_0)\right)=0\]
 
\noindent it holds 
\[\begin{array}{l}
\frac {\partial f}{\partial t}(t,r,z)+k^{sh}\left(b(r)-z\right)\frac {
\partial f}{\partial z}(t,r,z)+\frac 12\overline {\sigma}^2\left(
|r-z|^2\right)z\frac {\partial^2f}{\partial z^2}(t,r,z)\\
\qquad\qquad\quad\quad +\lambda (\pi ,r)\sum_{k=-m}^m\left[\psi\left
(t,r+k\delta ,z\right)-\psi\left(t,r,z\right)\right]p(\pi ,r,k\delta 
)-z\psi (t,r,z)\leq 0,\\
\qquad\qquad\qquad\qquad\qquad\qquad\qquad\qquad\qquad\qquad\qquad
0\leq t<\Theta ,~r\in (\underline r,\overline r
),\,\,\,z>0,\\
f(\Theta ,r,z)\geq\Psi (\pi ,r,z),\qquad\qquad\qquad\qquad\qquad\qquad
\qquad\qquad r\in (\underline r,\overline r),\,\,\,z>0.\end{array}
\]
\end{defi}

\noindent As mentioned in the Introduction, 
we will obtain existence and uniqueness of the viscosity 
solution to $(\ref{valeq})$ 
from a general result for valuation 
equations for contingent claims written 
on jump-diffusion underlyings proved in \cite{CostantiniPapiD'Ippoliti2012}. 
Since the state space of $(\ref{valeq})$ is unbounded, as 
usual in the literature uniqueness will hold in the class of functions with 
a prescribed growth rate.
For the convenience of the reader we summarise here the results of \cite{CostantiniPapiD'Ippoliti2012}. 
\cite{CostantiniPapiD'Ippoliti2012} considers a general equation of the form
\begin{equation}\label{eq-1}\left\{\begin{array}{ll}
\partial_t\psi(t,x)+L\psi(t,x)-c(x)\psi(t,x)=g(t,x),\quad&(t,x)\in (0,T)\times
D,\\
\psi(T,x)=\Psi (x),\quad&x\in D,\end{array}
\right.\end{equation}
with
\begin{equation}\label{op}Lf(x)=\nabla f(x)b(x)+\frac 12\mbox{\rm tr}\left(
\nabla^2f(x)a(x)\right)+\int_D\left[f(x')-f(x)\right]m(x,dx'),\end{equation}
where $D$ is a (possibly unbounded) starshaped open subset of $\R^d$.
The following assumptions on the coefficients will be 
made.

\begin{itemize}
\item[(H1)]\label{H1}
$a:D\rightarrow\R^{d\times d}$ is of the form $a=\sigma\sigma^T$, 
with $a=(a_{i,j})_{i,j=1,\ldots ,d}$, where $a_{i,j}\in {\cal C}^
2(D)$, 
and $b:D\rightarrow\R^d$ is Lipschitz continuous on compact subsets of $
D$.
\item[(H2)]\label{H2}
Denoting by ${\cal M}(D)$ the space of finite Borel measures on D, endowed with 
the weak convergence topology,
$m:D\rightarrow {\cal M}(D)$ is continuous and
\begin{eqnarray}
\sup_{x\in D}\left|\int_Df(x')m(x,dx')\right|<\infty ,\qquad\,\,\,
\forall\;f\in {\cal C}_c(D).\end{eqnarray}
\item[(H3)]\label{H3}
There exists a nonnegative function $V\in {\cal C}^2(D)$, such that
\begin{eqnarray*}
&&\int_DV(x')m(x,dx')<+\infty ,\,\;\;\forall\;x\in D,\quad LV(x)\leq 
C\left(1+V(x)\right),\;\;\forall\;x\in D,\\
&&\lim_{x\in D,\,x\rightarrow x_0}V(x)=+\infty ,\,\forall\;x_0\in 
D,\quad\lim_{x\in D,\,|x|\rightarrow +\infty}V(x)=+\infty ,\end{eqnarray*}
and
\item[(H4)]\label{H4} $g\in {\cal C}([0,T]\times D)$, $c,\,\psi\in 
{\cal C}(D)$, and
$c$ is bounded from below. There exists a strictly increasing function $
l:[0,+\infty )\rightarrow [0,+\infty )$, such that
\begin{eqnarray*}
&&s\mapsto sl(s)\;\quad\mbox{\rm is convex,}\qquad\lim_{s\rightarrow 
+\infty}l(s)=+\infty ,\\
&&(s_1+s_2)l(s_1+s_2)\leq C\left(s_1l(s_1)+s_2l(s_2)\right),\quad\;\;
\forall\;s_1,s_2\geq 0,\end{eqnarray*}
and the following holds:
\begin{eqnarray*}
|g(t,x)|l(|g(t,x)|)&\leq&C_T\left(1+V(x)\right),\\
|\Psi (x)|l(|\Psi (x)|)&\leq&C\left(1+V(x)\right),\end{eqnarray*}
for all $(t,x)\in [0,T]\times D$.
\end{itemize}

\begin{teo}\label{CPD}
Assume that \ref{H1},\ref{H2},\ref{H3} and \ref{H4} hold.
Then for every probability distribution $P_0$ on $D$, there exists
one and only one stochastic process $X$ solution of the martingale problem for
$(L,P_0)$ with $\mathcal D(L)=\mathcal C^2_c(D)$, that is the homogeneous strong 
Markov process $X$ with right continuos paths with left hand limits.

\noindent Moreover, denoting by $X_x$ the process with $P_0=\delta_x$, 
we have, for every $x\in D$, and for every 
$t\in [0,\Theta ]$,

\begin{equation*} 
\E\left[\,\left|\Psi (X_x(\Theta -t))e^{-\int_0^{\Theta -
t}c(X_x(r)dr}-\int_0^{\Theta -t}g(t+s,X_x(s)e^{-\int_0^sc(X_x(r)dr}ds\right
|\,\,\right]<+\infty .\label{u-int}
\end{equation*}

\noindent and the function
\begin{equation}\psi (t,x)=\E\left[\Psi (X_x(\Theta -t))e^{-\int_
0^{\Theta -t}c(X_x(r)dr}-\int_0^{\Theta -t}g(t+s,X_x(s)e^{-\int_0^
sc(X_x(r)dr}ds\,\right],\label{fk}\end{equation}
\vskip.1in
\noindent 
for all $(t,x)\in [0,\Theta ]\times D$, is continuous on $[0,\Theta ]\times D$ and is the only viscosity solution of $
(\ref{eq-1})$  satisfying
\begin{equation}|\psi (t,x)|l(|\psi (t,x)|)\leq C_{\Theta}\left(1
+V(x)\right),\;\qquad\forall\;(t,x)\in [0,\Theta]\times D.\label{V-growth}\end{equation}
\end{teo}\ 

\noindent
We are now ready to state the main result of this 
section. 

\begin{teo}\label{sol}
Let $\varphi$ be the function defined by $(\ref{price})$ and let 
$\varphi^i$ be the functions defined in Lemma \ref{recursion}. 
Then: For each $\pi\in\R$, The function $\varphi^{M-1}(t_{M-1}+\cdot 
,\pi ,\cdot ,\cdot )$ 
is the unique viscosity solution of the equation 
$(\ref{valeq})$ with terminal condition $\Psi^{M-1}:=B\Phi$ 
satisfying $(\ref{V-growth})$, where $V$ is defined by 
$(\ref{V})$-$(\ref{V0})$-$(\ref{V1})$ below. 
For each $i=0,...,M-2$, for each $\pi\in\R$, 
the function $\varphi^i(t_i+\cdot ,\pi ,\cdot ,\cdot )$ 
is the unique viscosity solution of the equation 
$(\ref{valeq})$ with terminal condition 
$\Psi^i:=B\left(\varphi^{i+1}(t_{i+1},\cdot ,\cdot ,\cdot )\right
)$ satisfying $(\ref{V-growth})$ with 
$V$ as above.  
\end{teo}

\begin{proof}
In order to adjust to the general formulation of 
\cite{CostantiniPapiD'Ippoliti2012}, we note first of all that $(\ref{valeq})$ can 
be viewed as a simple Partial Integro-Differential 
Equation, namely 
\begin{equation}\begin{array}{l}
\frac {\partial\psi}{\partial t}(t,r,z)+k^{sh}\left(b(r)-z\right)\frac {
\partial\psi}{\partial z}(t,r,z)+\frac 12\overline {\sigma}^2\left
(|r-z|^2\right)z\frac {\partial^2\psi}{\partial z^2}(t,r,z)\\
\qquad\qquad\quad\quad +\int_{(\underline r,\overline r)}\left[\psi\left
(t,r',z\right)-\psi\left(t,r,z\right)\right]\mu (\pi ,r,dr')-z\psi 
(t,r,z)=0,\\
\qquad\qquad\qquad\qquad\qquad\qquad\qquad\qquad\qquad\qquad 0\leq 
t<\Theta ,~r\in (\underline r,\overline r),\,\,\,z>0,\\
\psi (\Theta ,r,z)=\Psi^i(\pi ,r,z),\qquad\qquad\qquad\qquad\qquad 
r\in (\underline r,\overline r),\,\,\,z>0,\end{array}
\label{PIDE}\end{equation}
where 
\begin{equation}\mu (\pi ,r,A):=\lambda (\pi ,r)\sum_{k=-m,\,\,k\neq 
0}^m\1_A(r+\delta k)p(\pi ,r,k\delta ),\qquad A\in {\cal B}\big((\underline 
r,\overline r)\big).\label{PIDE-mu}\end{equation}
Therefore, for each fixed $\pi$, $(\ref{PIDE})$ is of the form \ref{eq-1}-\ref{op} 
with 
\[\begin{array}{ll}
L^{\pi}\psi (t,r,z)=k^{sh}\left(b(r)-z\right)\frac {\partial\psi}{
\partial z}(t,r,z)+\frac 12\overline {\sigma}^2\left(|r-z|^2\right
)z\frac {\partial^2\psi}{\partial z^2}(t,r,z)\\
\qquad\qquad\qquad\qquad\qquad +\int_{(\underline r,\overline r)}\left
[\psi\left(t,r',z\right)-\psi\left(t,r,z\right)\right]\mu (\pi ,r
,dr'),\end{array}
\]
\[g(r,z)=0,\qquad c(r,z)=z.\]

Assumptions \ref{H1},\ref{H2} of theorem \ref{CPD} are satisfied by 
$(\ref{continuitapelambda})$ and 
$(\ref{sigmabarproperty1})$. 

As far as Assumption \ref{H3} is concerned, it is sufficient to find, for 
each fixed $\pi$, 
$V_0^{\pi}\in {\cal C}^2(\underline r,\overline r)$, $V_1^{\pi}\in 
{\cal C}^2((0,\infty ))$ nonnegative and such that 
\begin{equation}\lim_{r\rightarrow\underline r^{+}}V_0^{\pi}(r)=\infty 
,\quad\lim_{r\rightarrow\overline r^{-}}V_0^{\pi}(r)=\infty ,\quad 
L^{\pi}V_0^{\pi}(r,z)\leq C(1+V_0^{\pi}(r)),\label{V0cond}\end{equation}
\begin{equation}\lim_{z\rightarrow 0^{+}}V_1^{\pi}(z)=\infty ,\quad\lim_{
z\rightarrow\infty}V_1^{\pi}(z)=\infty ,\quad L^{\pi}V_1^{\pi}(r,
z)\leq C(1+V_1^{\pi}(z)).\label{V1cond}\end{equation}
Then Assumption \ref{H3} will be verified by 
\begin{equation}V^{\pi}(r,z):=V_0^{\pi}(r)+V_1^{\pi}(z).\label{V}\end{equation}
By the same computations as in Theorem \ref{theowellpos}, 
and by $(\ref{sigmabarproperty2})$,  
$(\ref{bandsigmaproperty})$, we can see that 
\begin{equation}V^{\pi}_1(z):=z^2-\ln(z).\label{V1}\end{equation}
satisfies $(\ref{V1cond})$.
For each fixed $\pi$, $V^{\pi}_0$ can be constructed in the following 
way. By $(\ref{derivatefinitepelambda})$, there exist $\underline {
\beta}\pi=\underline {\beta} ,\overline {\beta}\pi =
\overline {\beta}\in {\cal C}^1([\underline r,\overline r]
)$ such that 
\[\underline {\beta}(\underline r)=0,\quad\underline {\beta}(r)>0\mbox{\rm \ for }
r>\underline r,\quad\underline {\beta}\mbox{\rm \ is nondecreasing}
,\]
\[\underline {\beta}(r)\geq\max_{h=1,..,m}p(\pi ,r+h\delta ,-h\delta 
),\]
and 
\[\overline {\beta}(\overline r)=0,\quad\overline {\beta}(r)>0\mbox{\rm \ for }
r<\,\overline r,\quad\overline {\beta}\mbox{\rm \ is nonincreasing}
,\]
\[\overline {\beta}(r)\geq\max_{h=1,..,m}p(\pi ,r-h\delta ,h\delta ).\]
Setting 
\begin{equation}V^\pi _0(r):=\int_r^{\overline r}\frac 1{\underline {\beta}(s)}\,ds\,+\,
\int_{\underline r}^r\frac 1{\overline {\beta}(s)}\,ds,\label{V0}\end{equation}
we have, for $k=1,...,m$, $r-k\delta\in (\underline r,\overline r
)$, 
\begin{eqnarray*}
&&\lambda (\pi ,r)\left[V^\pi_0(r-k\delta )-V^\pi_0(r)\right]p(\pi ,r,-k\delta 
)\\
&=&\lambda (\pi ,r)\left[\int_{r-k\delta}^r\frac 1{\underline {\beta}
(s)}\,ds\,+\,\int_{r-k\delta}^r\frac 1{\overline {\beta}(s)}\,ds\right
]p(\pi ,r,-k\delta )\\
&\leq&\overline {\lambda}\,\frac {m\delta}{\underline {\beta}(r-k
\delta )}p(\pi ,r,-k\delta )\,+\overline {\lambda}V^\pi_0(r)\\
&\leq&\overline {\lambda}(m\delta +1)(1+V^\pi_0(r)).\end{eqnarray*}
Analogously,  for $k=1,...,m$, $r+k\delta\in (\underline r,\overline 
r)$, 
\begin{eqnarray*}
&&\lambda (\pi ,r)\left[V^\pi_0(r+k\delta )-V^\pi_0(r)\right]p(\pi ,r,k\delta 
)\\
&=&\lambda (\pi ,r)\left[\int_r^{r+k\delta}\frac 1{\underline {\beta}
(s)}\,ds\,+\,\int_r^{r+k\delta}\frac 1{\overline {\beta}(s)}\,ds\right
]p(\pi ,r,k\delta )\\
&\leq&\overline {\lambda}V^\pi_0(r)+\overline {\lambda}\,\frac {m\delta}{\overline {
\beta}(r+k\delta )}p(\pi ,r,k\delta )\,\\
&\leq&\overline {\lambda}(m\delta +1)(1+V^\pi_0(r)).\end{eqnarray*}
Thus $(\ref{V0cond})$ is satisfied. 
 
We now turn to Assumption \ref{H4} of theorem \ref{CPD}. 
The conditions on $c(z):=z$ and $g(z)\equiv 0$ are clearly satisfied 
(note that the lower bound of $c$ in 
\cite{CostantiniPapiD'Ippoliti2012} does not need to be positive). 
Taking the function $l$ of 
\cite{CostantiniPapiD'Ippoliti2012} as $l(q)=q$,  
any continuos terminal value with sublinear growth 
satisfies the conditions of Assumption \ref{H4}. 
Therefore Assumption \ref{H4} is satisfied by Lemmas \ref{param-prop} 
and \ref{B-prop}. 

Since $(R^{\pi}_{r},R^{sh,\pi}_{(r,z)})$ is a solution of the martingale problem for $(L^{\pi},\delta_{(r,z)})$.
The thesis follows from Theorem \ref{CPD}. 
\end{proof}

\section{Numerical Simulation}\label{sectionTheImplementation}

In this section we present a finite difference scheme to solve numerically the pricing problem of an interest-rate financial derivative under our model of Section 2. Let us remind that in recent years a great deal has been done for the numerical approximation of viscosity solutions for second order problems. In particular, we refer the reader to the fundamental paper by Barles and Souganidis \cite{BarlesSouganidis1991} who first showed convergence results for a large class of numerical schemes to the solution of fully nonlinear second order elliptic or parabolic partial differential equations. Moreover we refer to \cite{BrianiNataliniChioma2004} for the extension of their arguments to the class of numerical schemes for integro-differential equations.

Our numerical scheme is applied to the sequence of partial differential equations and their final condition (\ref{valeq})-(\ref{terminal}), without using any artificial condition at the boundary $z=0$. The effectiveness of the method is based on the use of assumption (\ref{bandsigmaproperty}) and is tested on some numerical examples.\\

For each interval $(t_i,t_{i+1})$ and for every fixed value for the inflation rate $\pi$, we calculate the solution of problem (\ref{valeq})-(\ref{terminal}) by the following method. We convert the problem into an initial-value problem letting $\tau=\Theta-t$, hence we define $\psi^{n}_{h,j}$ as the approximate value of the solution $\psi$ at $\Theta-\tau_n$, $\tau_n=n\Delta \tau$, $n=0,\ldots, N$, $r_h=\underline{r}+h\delta$, $h=1,\ldots,H-1$, $z_j=j\Delta z$, $j=0,\ldots,J$, where $\Delta\tau=\Theta /N$, $\Delta z=z_{\max}/J$, $N$, $H$, $J$, being positive integers, such that $H\leq (\overline{r}-\underline{r})/ \delta$, and $z_{\max}>>0$. Therefore, the numerical domain of the problem is $[0,\Theta]\times [0,z_{\max} ]$. Given the discrete nature of the ECB rate, we propose an approximation of the solution at $r_h$, where the increments are proportional to the minimum jump size $\delta$. In fact, this choice allows to deal effectively with the non-local term in the equation (\ref{valeq}), using linear algebraic equations.

At a point $(\tau_n, r_h,z_j)$, $z_j>0$, the partial differential equation in (\ref{valeq}) can be discretized by the following second order approximation:

\begin{eqnarray}\label{num.schema.int}
&&\frac{\psi^{n+1}_{h,j}-\psi^{n}_{h,j}}{\Delta \tau}=\frac{k^{sh}(b(r_h)-z_j)}{4\Delta z}\left(\psi^{n+1}_{h,j+1}-\psi^{n+1}_{h,j-1}+\psi^{n}_{h,j+1}-\psi^{n}_{h,j-1}\right)+\nonumber\\
&&\frac{1}{4\Delta z^2}z_j\overline{\sigma}^2(|r_h-z_j|^2)\left(\psi^{n+1}_{h,j+1}-2\psi^{n+1}_{h,j}+\psi^{n+1}_{h,j-1}+\psi^{n}_{h,j+1}-2\psi^{n}_{h,j}+\psi^{n}_{h,j-1}\right)\nonumber\\
&&+\lambda(\pi,r_h)\sum_{k=-\min(m,h-1)}^{\min(m,H-h-1)}\left[\psi^{n}_{h+k,j}-\psi^{n}_{h,j}\right]p(\pi,r_h,k\delta)-\frac{z_j}{2}\left(\psi^{n+1}_{h,j}+\psi^{n}_{h,j}\right),
\end{eqnarray}
for any $h=1,\ldots,H-1$, $j=1,\ldots J-1$, $n=0,\ldots, N-1$. At the boundary $z=0$, the partial differential equation 
in (\ref{valeq}) becomes a hyperbolic equation with respect to $z$, with a nonlocal term:
\begin{eqnarray}\label{hyp.eq}
\frac{\partial \psi}{\partial t}(t,r,z)+k^{sh}b(r)\frac{\partial \psi}{\partial z}(t,r,z)+\lambda(\pi,r)\sum_{k=-m}^m\left[\psi(t,
r+k\delta,z)-\psi(t,r,z)\right]p(\pi,r,k\delta)=0\nonumber\\
\end{eqnarray}
Since $b(r)>0$, the value of $\psi$ on the boundary $z=0$ should be determined by the value of $\psi$ inside the domain. Hence, we can approximate the partial differential equation by the following scheme:
\begin{eqnarray}\label{num.s.b}
&&\frac{\psi^{n+1}_{h,0}-\psi^{n}_{h,0}}{\Delta\tau}=\frac{k^{sh}b(r_h)}{4\Delta z}\left(-\psi^{n+1}_{h,2}+4\psi^{n+1}_{h,1}-3\psi^{n+1}_{h,0}-\psi^{n}_{h,2}+4\psi^{n}_{h,1}-3\psi^{n}_{h,0}\right)\nonumber\\
&&+\lambda(\pi,r_h)\sum_{k=-\min(m,h-1)}^{\min(m,H-h-1)}\left[\psi^{n}_{h+k,0}-\psi^{n}_{h,0}\right]p(\pi,r_h,k\delta),
\end{eqnarray}
for any $h=1,\ldots,H-1$, $n=0,\ldots, N-1$. Here $\partial \psi/\partial z$ is discretized by a one-side second order scheme in order for all the node points involved to be in the computational domain. Moreover we assign the initial
datum at $\psi^{0}_{h,j}=\psi(\Theta,r_h,z_j)=\Psi^i(\pi,r_h,z_j)$, for any $j=0,\ldots,J$. At the boundary $z=z_{\max}$ we adopt 
the Neumann boundary condition $\psi^{n}_{h,J}=\psi^{n}_{h,J-1}$, for any $n=0,\ldots,N$.
When $\psi^{n}_{h,j}$, $h=1,\ldots,H-1$, $j=0,\ldots,J$ are known from (\ref{num.schema.int}) and (\ref{num.s.b}), we can determine $\psi^{n+1}_{h,j}$, for any $h$ and $j$. Therefore, we can perform this procedure for $n=0,\ldots,N-1$ successively and finally find $\psi^{N}_{h,j}$, for any $h$ and $j$. Since truncation errors are second order everywhere, at least for a smooth enough solution it may be expected that the error is $O(\Delta \tau^2, \Delta z^2)$, see ~\cite{Marcozzi2001} and ~\cite{ZhuLi2003}.  In order to evaluate the numerical solution at the time step $n+1$, we rewrite equations (\ref{num.schema.int}) and (\ref{num.s.b}) throughout using the following quantities:
\begin{eqnarray}\label{coeff.nu}
\nu_{h,j}=\frac{k^{sh}}{4}(b(r_h)-z_j)\frac{\Delta \tau}{\Delta z},\qquad\qquad h=1,\ldots, H-1,\;\;\;\ j=0,\ldots J-1, 
\end{eqnarray}
\begin{eqnarray}\label{coeff.nu}
\xi_{h,j}&=&\frac{z_j}{4}\overline{\sigma}^2(|r_h-z_j|^2)\frac{\Delta \tau}{(\Delta z)^2},\qquad\qquad h=1,\ldots, H-1,\;\;\;\ j=1,\ldots J-1, \\
\xi_{h,0}&=& \frac{3}{4}k^{sh}b(r_h)\frac{\Delta \tau}{\Delta z}.
\end{eqnarray}
\begin{eqnarray}\label{coeff.eta}
\eta_{h,j}=\xi_{h,j}+\nu_{h,j},\qquad\;\; \theta_{h,j}=\nu_{h,j}-\xi_{h,j},\qquad\;\; w_{h,j}=2\xi_{h,j}+\frac{\Delta\tau}{2}z_j+1.
\end{eqnarray}
Moreover, for every $n=0,\ldots, N-1$, $h=1,\ldots, H-1$, $j=1,\ldots, J-1$, we define
\begin{eqnarray}\label{coeff.Q}
Q^{n}_{h,j}&=&\psi^n_{h,j}+\nu_{h,j}\left(\psi^n_{h,j+1}-\psi^n_{h,j-1}\right)+\xi_{h,j}\left(\psi^n_{h,j+1}-2\psi^n_{h,j}+\psi^n_{h,j-1}\right)\nonumber\\
&&+\Delta\tau \lambda(\pi,r_h)\sum_{k=-\min(m,h)}^{\min(m,H-h)}\left[\psi^n_{h+k,j}-\psi^n_{h,j}
\right]p(\pi,r_h,k\delta)-\frac{z_j\Delta \tau}{2}\psi^n_{h,j},\\ 
Q^{n}_{h,0}&=&\psi^n_{h,0}+\nu_{h,0}\left(-\psi^n_{h,2}+4\psi^n_{h,1}-3\psi^n_{h,0}\right)+\nonumber\\
&&+\Delta\tau \lambda(\pi,r_h)\sum_{k=-\min(m,h)}^{\min(m,H-h)}\left[\psi^n_{h+k,0}-\psi^n_{h,0}
\right]p(\pi,r_h,k\delta).
\end{eqnarray}
\begin{eqnarray}\label{matrix}
\!A_h\!=\!\left[\begin{array}{cccccc}
1+\xi_{h,0} & -4 \nu_{h,0} & \nu_{h,0} & 0 & \cdots & 0\\
\theta_{h,1}& w_{h,1} & -\eta_{h,1} & 0 & \cdots & 0 \\
0 & \theta_{h,2} & w_{h,2} & -\eta_{h,2}& \cdots & 0\\
\vdots &  \vdots &  \vdots &  \vdots &  \vdots &  \vdots  \\
0 & 0 & 0 & 0 & \theta_{h,J-1} & (w_{h,J-1} -\eta_{h,J-1})
\end{array}\right]\;
K_h^n=\left[\begin{array}{c}
Q^{n}_{h,0}\\
Q^{n}_{h,1}\\
\vdots\\
Q^{n}_{h,J-1}\\
\end{array}\right],
\end{eqnarray}
$A_h$ is a $J\times J$ matrix independent of $\psi^{n+1}_{h,.}$ and $\psi^{n}_{h,.}$, whereas $K_h^n\in \mathbb{R}^J$ depends on the values of the numerical solution at the time step $n$. Therefore, keeping the terms 
which involve $\psi^{n+1}_{h,j}$, for $j=0,\ldots,J-1$, on the left-hand side of equation (\ref{num.schema.int}), (\ref{num.s.b})  and bringing all the other terms on the right-hand side, we easily obtain the following linear system:
\begin{eqnarray}\label{LS.1}
A_h \psi^{n+1}_h=K_h^n
\end{eqnarray}
for the computation of the numerical solution at the time step $n+1$, given by
\begin{eqnarray}\label{LS.2}
\psi^{n+1}_{h}=\left[\begin{array}{c}
\psi^{n+1}_{h,0}\\
\psi^{n+1}_{h,1}\\
\vdots\\
\psi^{n+1}_{h,J-1}\\
\end{array}\right],
\end{eqnarray}
for any $h=1,\ldots,H-1$. We observe that the coefficients in (\ref{coeff.eta}) satisfy
\begin{eqnarray}\label{LS.3}
w_{h,j}>|\theta_{h,j}|+\eta_{h,j},\qquad\qquad \mbox{\rm for all $j=1,\ldots J-1$},
\end{eqnarray}
and the same holds for the coefficients in the first row of $A_h$. Therefore $A_h$ is strictly diagonally dominant, implying that $A_h$ is invertible; moreover, since $w_{h,j}>1$, for any $j=1,\ldots,J-1$, 
the real parts of its eigenvalues are positive. In fact, we recall that these results follow from the well known Gershgorin's circle theorem. Therefore system (\ref{LS.1}) admits a unique solution.\\
For each discretized value $\pi$ of the observed inflation rate 
at time $t_i$, the numerical procedure allows to obtain 
$\psi^N_{h,j}=\psi^N_{h,j}(\pi )$, i.e. the approximate value of $
\varphi^i(t_i,\pi ,r_h,z_j)$, 
for any $h=1,\ldots ,H-1$, $j=0,\ldots ,J$, from the initial datum 
$\Psi^i$ evaluated at $(\pi ,r_{h'},z_{j'})$, $h'=1,\ldots ,H-1$, $
j'=0,\ldots ,J$. For 
$i=M-1,$ for each discretized value of $\pi$ and for each 
$h'=1,\ldots ,H-1$, $j'=0,\ldots ,J$, $\Psi^{M-1}(\pi ,r_{h'},z_{
j'})$ is obtained 
from the payoff $\Phi$ by applying a standard quadrature 
method for the evaluation of the integral operator $B$ 
defined in $(\ref{B})$. For $i<M-1$, $\Psi^i(\pi ,r_{h'},z_{j'})$ is 
obtained analogously from the approximate values of 
$\varphi^{i+1}(t_{i+1},\pi',r_{h'},z_{j'})$, where $\pi'$ ranges over all discretized 
values of the inflation rate (the grid for the variable $u$ 
in the integral operator $B$ can be chosen so that 
$\gamma (\pi ,r_{h'},z_{j'})+u$ is corresponds to a discretized value of the inflation 
rate). 

In the following section we present a numerical experiment using the finite difference scheme  
defined above. We remark that, in the case that the payoff function, the intensity function $\lambda$ and the probability for jump occurrence $p$ are independent of $\pi$, then the solution to the pricing problem is independent of $\pi$ and the application of the operator $B$ is not needed; the numerical scheme (\ref{num.schema.int})-(\ref{num.s.b}) can be applied once with $\Theta$ replaced by $M\Theta$.

\subsection{The Numerical Tests}\label{NS1}

In this section we present a numerical experiment related to a pricing problem with a specific payoff function (see \ref{payoff.test.1} below), using the numerical scheme described above. The diffusion $\overline{\sigma}$, the jump-intensity $\lambda$ are chosen as constant 
functions. The same holds for the probability of jumps $p(\pi,r,k\delta)=1/(2m+1)$, if $r+k\delta\in (\underline{r},\overline{r})$, zero otherwise, and all the interest rates as well as the inflation rate are 
expressed as a percentage. The unit time corresponds to one year. The remaning coefficients in the model are fixed as in Table \ref{T_parameter}.
\begin{table}[]
{\small
  \begin{center}
     \begin{tabular}{|c|c|c|c|c|} \hline
$\overline{\sigma}=0.23$ & $\lambda=10$ & $k^{sh}=2$ & $\underline{r}=0.05 $ & $\overline{r}=4.25$\\
\hline
$b_0=0$ & $b_1=1$ & $\delta=0.25$ & $m=4$ & $\Theta=0.25$ \\
\hline
$\pi=1.52$ & $\alpha=0.8$ & $\beta=1.2$ & $k_\pi=3$ & $\pi^{\star}=2$ \\
\hline
$v=0.1$ & $M=4$ & $\beta=1.2$ & $k_\pi=3$ & $\pi^{\star}=2$ \\
\hline
\end{tabular}
\end{center}}
\caption{Model parameters for numerical tests }\label{T_parameter}
\end{table}
\noindent 
We observe that $\underline{r}$ and $\overline{r}$ are fixed as the minimum and the maximum
historical value of the ECB rate observed between July 2008 to September 2014, respectively. 
Since $\Theta=0.25$, the inflation rate is observed once every three months, and $\pi^{\star}$ is fixed as the target inflation rate established in the Eurozone. 

\noindent 
We conduct our numerical test in the case of the inflation-indexed swaps (IIS), where the standard no-arbitrage pricing theory implies that the value at time $t<T$ of the inflation-indexed leg of the IIS can be expressed as
\begin{eqnarray}\label{payoff.test.1}
{\cal N} \mathbb{E}\left[ e^{-\int_t^T R^{sh}(s) ds}\left(\frac{\Pi(T)}{\Pi_0}-1\right)\Bigg{|} {\cal F}_t\right],
\end{eqnarray}
where $T$ is the maturity date, ${\cal N}$ is the nominal value of the contract, $\Pi_0$ is a reference value for the inflation rate. Hence the associated payoff function to (\ref{payoff.test.1}) is $\Phi(\pi,r,z)={\cal N}(\pi/\pi_0-1)$. In particular we refer the reader to \cite{Mercurio2005} for a detailed description of these kind of inflation-linked instruments. Here, for the sake of simplicity, we set ${\cal N}=\Pi_0=1$. 

\noindent
We present the error analysis taking into the order $\gamma$ of the numerical error under the
discrete $l^1$-norm and $l^\infty$-norm,
\begin{equation}\label{gamma}
\gamma_{1,\infty}=\log_2\left(\displaystyle\frac{e^{N,J}_{1,\infty}}{e^{2N,2J}_{1,\infty}}\right),
\end{equation}
where the relative $l^1$ and $l^\infty $ errors are calculated as follows
respectively as follows:
\begin{equation}\label{l1_error}
e^{N,J}_{1}=\max_h e^{N,J}_{1,h},\qquad \qquad
e^{N,J}_{1,h}=\displaystyle\frac{\displaystyle\sum_{j=0}^J |\varphi_{h,j}^{N,J}-\varphi_{h,2j}^{2N,2J}|}{\displaystyle\sum_{j=0}^J|\varphi_{h,2j}^{2N,2J}|},
\end{equation}
\begin{equation}\label{l_inf_error}
e^{N,J}_{\infty}=\max_h e^{N,J}_{\infty,h},\qquad \qquad
e^{N,J}_{\infty,h}=\displaystyle\frac{\displaystyle\max_{j=0}^J |\varphi_{h,j}^{N}-\varphi_{h,2j}^{2N,2J}|}{\displaystyle\max_{j=0}^J|\varphi_{h,2j}^{2N,2J}|}.
\end{equation}
Here $\varphi_{h,j}^{N,J}$ corresponds to the numerical solution approximating the value of the pricing function $\varphi$ at $r=r_h$, $z=z_j$, $t=0$, when the number of time steps in each iteration of the finite difference scheme (\ref{num.schema.int})-(\ref{num.s.b}) is equal to $N$ and the number of points in the interest rate grid is $J$. The total number of values for the ECB rate in the grid (=$H$) is given by the integer part of $(\overline{r}-\underline{r})/\delta$. We observe that, in the computation of the errors we compare the solution $\varphi_{h,j}^{N,J}$ with $\varphi_{h,2j}^{2N,2J}$ which corresponds to doubling the number of points both in the interest rate grid and in the time grid. Note that the $j$th interest rate value, when the number of points in the grid is $J$, coincides with the $(2j)$th interest rate value in the grid, when the total number of point is $2J$. Hence, we expect their difference to be sufficiently small as $N$ and $J$ increase. Moreover in Figure \ref{Error.fig} is showed the behavior of the $l^1$ error 
$\{e_1^{N,J}(h)\}_h$ and the $l^\infty$ error $\{e_\infty^{N,J}(h)\}_h$ as a function of the ECB rate $\{r_h\}_h$. The results of the error analysis are summarized in Table \ref{T} and it comes out that the scheme is of first order accuracy. The results we obtained are quite satisfactory and demonstrate the convergence of the numerical scheme. However, we plan to investigate with greater accuracy the convergence in a future work in preparation. Furthermore, the calibration to market data related to inflation-indexed derivatives is an interesting subject for future research.

\begin{table}[ht]
  \begin{center}
     \begin{tabular}{|c|c|c|c|c|} \hline
\multicolumn{5}{|c|}{$T=1$} \\ \hline
$N\times J$ & error $l^1$ & order $l^1$ & error $l^\infty$ & order $l^\infty$ \\
\hline 30$\times$30 & 0.0117 & 1.0298& 0.0116 & 1.2279\\ \hline
50$\times$50 &  0.0069& 1.0156 & 0.0063 & 1.3085\\
\hline 70$\times$70 &  0.0049& 1.0102 & 0.0041 & 1.3547\\ 
\hline 100$\times$100 & 0.0034 & 1.0067 & 0.0026 & 1.3454\\ 
\hline 150$\times$150 &  0.0023& 1.0042 & 0.0015 & 1.1660\\ 
\hline
\multicolumn{5}{|c|}{$T=2$} \\ \hline
$N\times J$ & error $l^1$ & order $l^1$ & error $l^\infty$ & order $l^\infty$ \\
\hline 30$\times$30 & 0.0419 & 1.0490& 0.0266 & 1.0832\\ \hline
50$\times$50 & 0.0244 & 1.0284 & 0.0151 & 1.0499\\
\hline 70$\times$70 &  0.0172& 1.0199& 0.0106& 1.0357\\ 
\hline 100$\times$100 &  0.0119& 1.0137& 0.0073& 1.0250\\ 
\hline 150$\times$150 &  0.0079& 1.0091& 0.0048 & 1.0166\\ 
\hline
\end{tabular}
 \end{center}
\caption{\small Numerical errors and order of accuracy of the scheme
with $z_{\max}=7$ (equivalent to the interest rate value of 7\%) for the maturity dates of 1 and 2 years, respectively. $N\times J$ stands for the mesh size.}\label{T}
\end{table}

\begin{figure}[h]
\vspace*{-1cm}
\begin{center}
 \includegraphics[width=12.5cm]{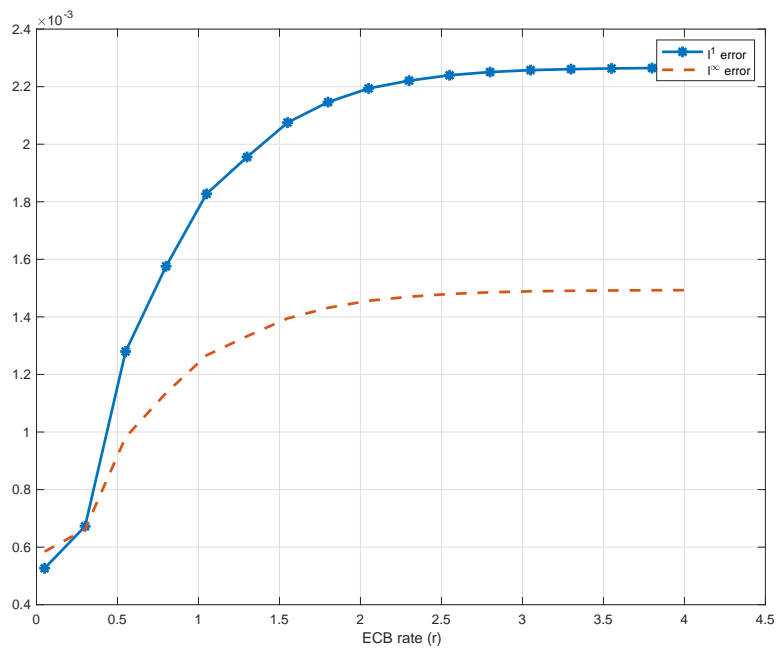}
\end{center}
\caption{$l^1$ error and $l^\infty$ error as a function of the ECB rate, for $N=J=50$, $T=2$.}\label{Error.fig}
\end{figure}

\vspace*{6.5cm}



\end{document}